\documentclass[draftcls,11pt,onecolumn,romanappendices]{IEEEtran}

\usepackage{amsthm,xpatch}
\usepackage{amsmath,amsfonts}
\usepackage{cite}
\usepackage{amssymb}
\usepackage{bm}
\usepackage{dsfont}
\usepackage{graphicx, subfigure}
\usepackage{breqn}
\usepackage{mathtools}
\usepackage{bbm}
\usepackage{latexsym}
\usepackage[ruled, linesnumbered]{algorithm2e}
\usepackage{accents}
\usepackage{soul}
\usepackage{comment}

\newcommand*{\transpose}{%
  {\mathpalette\@transpose{}}%
}

\usepackage{times}
\usepackage{amsmath}  
\usepackage{amssymb} 
\usepackage{mathrsfs}
\usepackage{cite}
\usepackage{comment} 
\usepackage{upref}
\usepackage{amsfonts}
\usepackage{verbatim}
\usepackage[dvipsnames,usenames]{color}
\usepackage{enumerate}
\usepackage{graphicx}
\usepackage{subfigure}
\usepackage{latexsym}
\usepackage{bm}
\usepackage{multirow}
\usepackage{dsfont}
\usepackage{amsthm,xpatch}
\usepackage{mathtools}
\usepackage{bbm}
\usepackage{latexsym}
\usepackage{accents}
\usepackage{psfrag}
\usepackage[font=footnotesize,justification=centering]{caption}
\usepackage{color}
\usepackage{times,color}
\usepackage[ruled,linesnumbered]{algorithm2e}
\usepackage[usenames,dvipsnames,svgnames,table]{xcolor}
\usepackage{pgf}
\usepackage{tikz}
\usetikzlibrary{arrows, automata, positioning, calc}
\usepackage{cancel} 
\usepackage{tikz, tikz-3dplot, amssymb}
\tdplotsetmaincoords{72}{130}
\makeatletter
\newif\if@borderstar
\def\bordermatrix{\@ifnextchar*{%
\@borderstartrue\@bordermatrix@i}{\@borderstarfalse\@bordermatrix@i*}%
}
\def\@bordermatrix@i*{\@ifnextchar[{\@bordermatrix@ii}{\@bordermatrix@ii[()]}}
\def\@bordermatrix@ii[#1]#2{%
\begingroup
\m@th\@tempdima8.75\p@\setbox\z@\vbox{%
\def\cr{\crcr\noalign{\kern 2\p@\global\let\cr\endline }}%
\ialign {$##$\hfil\kern 2\p@\kern\@tempdima & \thinspace %
\hfil $##$\hfil && \quad\hfil $##$\hfil\crcr\omit\strut %
\hfil\crcr\noalign{\kern -\baselineskip}#2\crcr\omit %
\strut\cr}}%
\setbox\tw@\vbox{\unvcopy\z@\global\setbox\@ne\lastbox}%
\setbox\tw@\hbox{\unhbox\@ne\unskip\global\setbox\@ne\lastbox}%
\setbox\tw@\hbox{%
$\kern\wd\@ne\kern -\@tempdima\left\@firstoftwo#1%
\if@borderstar\kern2pt\else\kern -\wd\@ne\fi%
\global\setbox\@ne\vbox{\box\@ne\if@borderstar\else\kern 2\p@\fi}%
\vcenter{\if@borderstar\else\kern -\ht\@ne\fi%
\unvbox\z@\kern-\if@borderstar2\fi\baselineskip}%
\if@borderstar\kern-2\@tempdima\kern2\p@\else\,\fi\right\@secondoftwo#1 $%
}\null \;\vbox{\kern\ht\@ne\box\tw@}%
\endgroup
}
\makeatother

\newtheorem{Corollary}{Corollary}
\newtheorem{theorem}{Theorem}
\newtheorem{remark}{Remark}
\newtheorem{proposition}{Proposition}
\newtheorem{definition}{Definition}
\newtheorem{conjecture}{Conjecture}

\makeatletter
\newcommand{\removelatexerror}{\let\@latex@error\@gobble}
\makeatletter
\newcommand{\proofpart}[2]{%
	\par
	\addvspace{\medskipamount}%
	\noindent\emph{Part #1: #2}\par\nobreak
	\addvspace{\smallskipamount}%
	\@afterheading
}
\makeatother
\makeatletter
\renewenvironment{proof}[1][\proofname]{\par
  \pushQED{\qed}%
  \normalfont \topsep6\p@\@plus6\p@\relax
  \trivlist
  \item[\hskip\labelsep
        \itshape
    #1\@addpunct{:}]\ignorespaces
}{%
  \popQED\endtrivlist\@endpefalse
}
\makeatother
\makeatletter
\renewcommand{\mathsf}[1]{#1}

\theoremstyle{definition}
\newtheorem{example}{Example}
\begin{document}

\title{\LARGE \textbf{\Large A Combinatorial View of the Service Rates of Codes Problem, \\[-0.75ex]
its Equivalence to Fractional Matching and its Connection with Batch Codes}}

\author{\normalsize {$^\ast$}Fatemeh Kazemi, {$^\ast$}Esmaeil Karimi, {$^\dagger$}Emina Soljanin, and {$^\ast$}Alex Sprintson\\
{\small {$^\ast$}ECE Dept., Texas A\&M University, USA \{fatemeh.kazemi, esmaeil.karimi, spalex\}@tamu.edu}\\ {\small {$^\dagger$}ECE Dept., Rutgers University, USA \{emina.soljanin\}@rutgers.edu}}

\maketitle

\thispagestyle{plain}

\begin{abstract} 
We propose a novel technique for constructing a graph representation of a code through which we establish a significant connection between the service rate problem and the well-known fractional matching problem. Using this connection, we show that the service capacity of a coded storage system equals the fractional matching number in the graph representation of the code, and thus is lower bounded and upper bounded by the matching number and the vertex cover number, respectively. This is of great interest because if the graph representation of a code is bipartite, then the derived upper and lower bounds are equal, and we obtain the capacity. Leveraging this result, we characterize the service capacity of the binary simplex code whose graph representation, as we show, is bipartite. Moreover, we show that the service rate problem can be viewed as a generalization of the multiset primitive batch codes problem.

\end{abstract}


\section{Introduction}

Providing reliability against failures, ensuring availability of stored content during high demand, providing fast content download and serving a large number of users simultaneously have always been major concerns in cloud storage systems. The service capacity has been very recently recognized as an important performance metric. It has a wide relevance, and can be interpreted as a measure of the maximum number of users that can be simultaneously served by a coded storage system~\cite{noori2016storage,aktacs2017service,anderson2018service,kazemi2020geometric,peng2018distributed,aktas2019load,allocation:sardariRFS10}. Thus, maximizing the service capacity is of great significance for the emerging applications such as distributed learning and fog computing. Moreover, maximizing the service capacity reduces the users' experienced latency, particularly in a high traffic regime, which is important for the delay-sensitive applications such as live streaming, where many users wish to get the same content at the same time. 

The service rate problem is concerned with a distributed storage system in which $k$ files $f_1,\dots,f_k$ are stored across $n$ servers using a linear ${[n,k]_q}$ code such that the requests to download file $f_i$ arrive at rate $\lambda_i$, and the server $l$ operates at rate $\mu_l$. A goal of the service rate problem is to determine the service rate region of this coded storage system which is the set of all request arrival rates $\boldsymbol{\lambda}=(\lambda_1,\dots,\lambda_k)$ that can be served by this system given the finite service rate of the servers. The service rate problem is generally formulated as a sequence of linear programs, that has been studied only in some limited cases\cite{aktacs2017service,anderson2018service,kazemi2020geometric}. In this paper, we show that the service rate problem is equivalent to the fractional matching problem which were extensively studied in the context of graph theory. This equivalence result allows one to leverage the techniques in the rich literature of the graph theory for solving the service rate problem.

\subsection{Previous and Related Work}

Existing studies on data access pursue various directions. Many are focused on providing efficient maintenance of storage under possible failures of a subset of nodes accessed (see e.g.,~\cite{dimakis2010network,dimakis2011survey,huang2013pyramid,gopalan2012locality,allocation:sardariRFS10}). These studies typically assume infinite service rate (instantaneous service) for each storage node. Hence, they do not address the problem of serving a large number of users simultaneously. 

Another important line of work is concerned with caching (see e.g.,~\cite{shanmugam2013femtocaching,maddah2016coding,hamidouche2014many}), in which generally the limited capacity of the backhaul link is considered as the main bottleneck of the system, and the goal is usually to minimize the backhaul traffic by prefetching the popular contents at the storage nodes of limited size. Thus, these works do not address the scenarios where many users want to get the same content concurrently given the limited capacity of the access part of the network. 

The other related body of work is concerned with minimizing the download latency (see e.g.,~\cite{joshi2012coding,shah2014mds,joshi2014delay,liang2014fast,gardner2015reducing,kadhe2015analyzing,kadhe2015availability,ISIT:AktasS18,aktas2017simplex,aktas2019analyzing,latency:JoshiSW15,latency:JoshiSW17}). These papers assume that the storage nodes can serve the customers at some finite rate, and aim to compute the download latency for intractable queueing systems that appear in coded storage. 

We note now and explain in detail later that because of the constraints on the service rate of servers, by maximizing the service capacity the load balancing is provided in the distributed storage system (see~\cite{aktas2019load}). In that sense, the most relevant work to this paper includes batch codes, switch codes and PIR codes (see e.g.,~\cite{ishai2004batch,skachek2018batch,lipmaa2015linear,fazeli2015pir,wang2017switch}). However, the problems considered in these papers, as we will show later, can be often seen as special cases of the service rate problem. 

A connection between distributed storage allocation problems (see~\cite{allocation:sardariRFS10,allocation:LeongDH12} and references therein) and matching problems in hyper-graphs have been observed in computer science literature~\cite{Matching:AlonFH12} (see also~\cite{Matching:KaoDLH13}). In particular, it was noted that the uniform model of distributed storage allocation considered in~\cite{allocation:sardariRFS10} leads to a question which is asymptotically equivalent to the fractional version of a long standing conjecture by Erd\H{o}s~\cite{erdHos1965problem} on the maximum number of edges in a uniform hypergraph.

\subsection{Main Contributions}

We first construct a special graph representation of a linear code in Sec.~\ref{graph-rep}. We then show the following results in Sec.~\ref{Equivalence-Results}:~1) equivalence between the service rate problem and the well-known fractional matching problem and 2) equivalence between the integral service rate problem and the matching problem. These equivalence results allow us to show that the service capacity of a code is equal to the fractional matching number in the graph representation of a code, and thus is lower bounded and upper bounded by the matching number and the vertex cover number, respectively. This is beneficial because if the graph representation of a code is a bipartite graph, then the upper bound and lower bound are equal, which allows us to establish the service capacity of the storage system. Leveraging this result, we determine the service capacity of the binary simplex codes whose graph representation, as we will show, is bipartite. Furthermore, we show that the service rate problem can be viewed as a generalization of batch codes problem in Sec.~\ref{generalization-batch}. 
In particular, we show that the multiset primitive batch codes problem is a special case of the service rate problem when the solution (the portion of requests assigned to the recovery sets) is restricted to be integral. \textit{All the proofs can be found in the Appendix.}

\section{Coded System and its Service Rate Region}\label{sec:Problem Formulation}

Throughout this work, we use bold-face lower-case letters for vectors and bold-face capital letters for matrices. Let $\mathbb{N}$ denote the set of positive integers. $\mathbb{F}_q$ denotes the finite field with $q$ elements. For $i \in \mathbb{N}$,  $[i]\triangleq\{1,\dots,i\}$. For $n \in \mathbb{N}$, $\mathbf{1}_n$ denotes the all-one vector of length $n$.

Consider a storage system where $k$ files $f_1,\dots,f_k$ are stored across $n$ servers labeled $1,\dots,n$, using an $[n,k]_q$ code with generator matrix $\mathbf{G} \in \mathbb{F}^{k\times n}_q$. A set of stored symbols that can be used to recover file $f_i$ is referred to as a recovery set for file $f_i$. Let $\mathbf{g}_j$ be the $j$th column of $\mathbf{G}$. The set $R\subseteq [n]$ is a recovery set for file $f_i$ if there exists non-zero $\alpha_j$'s $\in \mathbb{F}_q$ such that $\sum_{j \in R}\alpha_j \mathbf{g}_j=\mathbf{e}_i$, where $\mathbf{e}_i$ denotes the $i$th unit vector. In other words, a set $R$ is a recovery set for file $f_i$ if the unit vector $\mathbf{e}_i$ can be recovered by a linear combination of the columns of $\mathbf{G}$ indexed by the set $R$. 

Let ${t_i \in \mathbb{N}}$ denote the number of recovery sets for file $f_i$, and $\mathcal{R}_{i}=\{R_{i,1},\dots,R_{i,t_i}\}$ denote the set of recovery sets for file $f_i$. We assume w.l.o.g. that the time to download a file from server ${l\in [n]}$ is exponential with rate ${\mu_l \in \mathbb{R}_{\geq 0}}$, i.e., $\mu_l$ is the average rate at which server $l$ resolves the received file requests. We denote the service rates of servers $1,\dots,n$ by the vector ${\boldsymbol{\mu}=(\mu_1,\dots,\mu_n)}$. We further assume that the arrival of requests for file $f_i$ is Poisson with rate $\lambda_i$, $i\in [k]$. We denote the request rates for files $1,\dots,k$ by the vector ${\boldsymbol{\lambda}=(\lambda_1,\dots,\lambda_k)}$. We consider the class of scheduling strategies that assign a fraction of requests for a file to each of its recovery sets. Let $\lambda_{i,j}$ be the portion of requests for file $f_i$ that are assigned to the recovery set ${R_{i,j}}$, ${j\in [t_i]}$. 

The service rate problem seeks to determine the set of arrival rates $\boldsymbol{\lambda}=(\lambda_1,\dots,\lambda_k)$ that can be served by a coded storage system with generator matrix $\mathbf{G}$ and service rate $\boldsymbol{\mu}$, referred to as service rate region $\mathcal{S}(\mathbf{G},\boldsymbol{\mu}) \subseteq \mathbb{R}^k_{\geq 0}$.

\begin{definition}
An \ul{$(\mathbf{G},\boldsymbol{\mu})$ system} is a coded storage system in which $k$ files are stored across $n$ servers using a linear $[n,k]_q$ code with generator matrix $\mathbf{G} \in \mathbb{F}^{k\times n}_q$ such that file $f_i$ for $i \in [k]$ has $t_i \in \mathbb{N}$ recovery sets denoted by ${\mathcal{R}_{i}=\{R_{i,1},\dots,R_{i,t_i}\}}$, and the service rate of servers in the system is $\boldsymbol{\mu}=(\mu_1,\dots,\mu_n)$. 
\end{definition}

\begin{definition}\label{def:SerRateReg}
The \ul{service rate region} of an $(\mathbf{G},\boldsymbol{\mu})$ system, denoted by $\mathcal{S}(\mathbf{G},\boldsymbol{\mu})$, is the set of vectors ${\boldsymbol{\lambda}=(\lambda_1,\dots,\lambda_k)}$ for which there exist $\lambda_{i,j}$ satisfying the following constraints:
\begin{subequations}\label{eq:1}
\begin{align}\label{eq:1a}
&\sum_{j=1}^{t_i} \lambda_{i,j}=\lambda_i, ~~~~~~~~~~~ \text{for all} ~~~ i\in [k] \\
&\sum_{i=1}^{k}~\sum_{\substack{{j\in [t_i]} \\ \label{eq:1b} {l \in R_{i,j}}}} \lambda_{i,j} \leq \mu_l, ~~~ \text{for all} ~~~ l\in [n]\\
&\lambda_{i,j} \in \mathbb{R}_{\geq 0},~~~~~~~~~~~~~ \text{for all} ~~~ i\in [k],~j\in [t_i]\label{eq:pos}
\end{align}
\end{subequations}
\end{definition}

Note that constraints~\eqref{eq:1a} ensure that the demands for all files are served, and constraints~\eqref{eq:1b} guarantee that no node is sent requests in excess of its service rate. 

\begin{proposition}\cite[Lemma 1]{kazemi2020geometric}\label{prop:convexity}
The service rate region of an $(\mathbf{G},\boldsymbol{\mu})$ system $\mathcal{S}(\mathbf{G},\boldsymbol{\mu})$ is a non-empty, convex, closed, and bounded subset of the ${\mathbb{R}^k_{\geq 0}}$.
\end{proposition}

The service capacity of an $(\mathbf{G},\boldsymbol{\mu})$ system, $\lambda^\star(\mathbf{G},\boldsymbol{\mu})$, is defined as the maximum sum of arrival rates that can be served simultaneously by the storage system. We define a maximum demand vector, denoted by  ${\boldsymbol{{\lambda}}^\star=(\lambda_{1}^\star,\dots,\lambda_{k}^\star)}$, as a vector in the service rate region for which ${\sum_{i=1}^{k}\lambda_{i}^\star=\lambda^\star(\mathbf{G},\boldsymbol{\mu})}$. An instance of the maximum demand vector is obtained by solving the following linear programming (LP):
\begin{equation}\label{maxrate}
\max~~{\sum_{i=1}^{k} \lambda_i}
\qquad\text{s.t.}~~\eqref{eq:1}~~\text{holds}. 
\end{equation}

\begin{definition} The \ul{integral service rate region} of an $(\mathbf{G},\boldsymbol{\mu})$ system, denoted by $\mathcal{S}_{I}(\mathbf{G},\boldsymbol{\mu})$, is the set of all vectors ${\boldsymbol{\lambda}=(\lambda_1,\dots,\lambda_k)}$ for which there exist $\lambda_{i,j} \in \mathbb{Z}_{\geq 0}$ satisfying the sets of constraints $\eqref{eq:1a},\eqref{eq:1b}$.
\end{definition}

Note that each demand vector $\boldsymbol{\lambda}=(\lambda_1,\dots,\lambda_k)$ in the integral service rate region has integral coordinates, i.e., $\mathcal{S}_{I}(\mathbf{G},\boldsymbol{\mu}) \subseteq \mathbb{Z}^k_{\geq 0}$. However, because of the fractional relaxation of $\lambda_{i,j}$, it is not guaranteed that the vectors with integral coordinates in the service rate region $\mathcal{S}(\mathbf{G},\boldsymbol{\mu})$ are also in the integral service rate region $\mathcal{S}_{I}(\mathbf{G},\boldsymbol{\mu})$.

\begin{remark}
In the integral setting of the service rate problem where $\lambda_{i,j}$ are non-negative integers, if each server can serve up to one request at a time, i.e., $\mu_l=1$ for all servers $l \in [n]$, then one can easily conclude that $\lambda_{i,j}$ are binary and the recovery sets used for each file $f_i$, $i \in [k]$ are disjoint. 
\end{remark}

\section{Equivalence to Fractional Matching}

We first introduce a graph representation of a code which is useful for characterizing the service capacity of a coded storage system through relating this problem with the well-known problem of finding the maximum fractional matching in a graph. In particular, we show that the service capacity of a code equals the fractional matching number in our graph representation of the code. Another way of determining the service capacity of a coded storage system is providing tight bounds on the maximum sum of the arrival rates that can be served by the storage system. We show that the matching number and the vertex cover number in the graph representation of a code, respectively are a lower bound and an upper bound on the service capacity of a code. Thus, if the graph representation of a code is a bipartite graph, according to the Duality Theorem~\cite{scheinerman2011fractional}, the matching number and vertex cover number are identical, and we are able to determine the capacity. As an application of this result, we determine the service capacity of the binary simplex codes whose graph representation, as we will show, is a bipartite graph. We next describe how to construct the graph representation of a code, and then we present the interesting connections.

\subsection{Graph Representation of Codes}\label{graph-rep}

We focus on the settings with recovery sets of size $1$ and $2$ where the recovery sets for each file is either a systematic symbol or a group of two symbols. Extensions to the general case are mostly straightforward and involve hypergraphs in which each edge can be incident to an arbitrary number of vertices. The graph representation of a code with generator matrix $\mathbf{G}$ is denoted by $G(V,E)$ where the vertices in $V$ correspond to the $n$ encoded symbols (the servers of the storage system), and the edges in $E$ correspond to the recovery sets of files. In $G(V,E)$, each self-loop represents a recovery set of size $1$ for the vertex (file) that it is connected to, and each edge between two vertices represents a recovery set of size $2$ for the file that can be recovered from these two vertices. Each edge is assigned a color such that the edges that correspond to the recovery sets of the same file are assigned the same color.
In that sense, we have an edge-colored graph. It should be noted that a graph with self-loops can be simply converted to a graph without any self-loops by adding sufficient number of dummy vertices (servers). We assume that the label of all dummy servers is zero and thus we denote a systematic recovery set for file $f_i$ by $\{0,r\}$ where $r$ is the label of the systematic server storing file $f_i$. Section~\ref{example} provides an example that shows the graph representation of $[7,3]_2$ simplex code.

\subsection{Matching and Vertex Cover Problems~\cite{scheinerman2011fractional}}

\begin{definition}
A \ul{matching} in a graph is a set of all pairwise non-adjacent edges.
\end{definition}
Alternatively, a matching in a graph $G(V,E)$ is an assignment of the values $\tilde{x}_e \in \{0,1\}$ to the edges $e \in E$ in such a way that for each vertex $v \in V$, the sum of the values on the incident edges is at most $1$. All the edges $e \in E$ with value $\tilde{x}_e=1$ are in the matching. Thus, a \textit{matching vector} in a graph $G(V,E)$ can be defined as a vector ${\tilde{\boldsymbol{x}}=(\tilde{x}_e:e \in E)}$ satisfying the following conditions:
\begin{subequations}\label{match}
\begin{align}
&\sum_{\text{e incident to $v$}} \tilde{x}_e\leq 1, ~~~~~~~~ \text{for all} ~~ v \in V\\
&~~~~~\tilde{x}_e \in \{0,1\}, ~~~~~~~~~~~\text{for all} ~~ e \in E
\end{align}
\end{subequations}

\begin{definition}
A \ul{maximum matching} in a graph is a matching that contains the largest number of edges. The maximum matching vector is denoted by $\tilde{\boldsymbol{x}}^\star$.
\end{definition}
The size of a maximum matching in a graph $G(V,E)$ is called matching number, denoted by $m(G)$. There may be several instances of maximum matchings in a graph. The problem of finding an instance of maximum matching can be formulated as the following integer LP:
\begin{equation}\label{maxmatch}
\max~~{\sum_{e \in E}\tilde{x}_e}
\qquad\text{s.t.}~~\eqref{match}~~\text{holds}.
\end{equation}

\begin{definition}
A \ul{fractional matching} in a graph $G(V,E)$ is an assignment of the values ${x_e \in [0,1]}$ to the edges ${e \in E}$ such that for each vertex $v \in V$, the sum of the values on the incident edges is at most $1$.
\end{definition}
A \textit{fractional matching vector} in a graph ${G(V,E)}$ can be defined as a vector ${\boldsymbol{x}=(x_e:e \in E)}$ satisfying the following constraints:\begin{subequations}\label{FracMatch}
\begin{align}
&\sum_{\text{e incident to $v$}} x_e\leq 1, ~~~~~~~ \text{for all} ~~ v \in V\\
&~~~~~x_e \in [0,1], ~~~~~~~~~~~\text{for all} ~~ e \in E
\end{align}
\end{subequations}

\begin{definition}
A \ul{maximum fractional matching}, denoted by $\boldsymbol{x}^\star$, is a fractional matching vector in the graph that has the maximum value $\sum_{e \in E}x_e$ over all fractional matching vectors in the graph.
\end{definition}
The value of a maximum fractional matching in a graph $G(V,E)$ is called the fractional matching number, denoted as $m_f(G)$. Finding an instance of maximum fractional matching in a graph can be formulated as the following LP:\begin{equation}\label{MaxFracMatch}
\max~~{\sum_{e \in E}x_e}\\
\qquad\text{s.t.}~~\eqref{FracMatch}~~\text{holds}.
\end{equation}

\begin{definition}
A \ul{vertex cover} of a graph is a set of vertices such that each edge of the graph is incident to at least one vertex in the set.
\end{definition}
Alternatively, a vertex cover of a graph $G(V,E)$ is an assignment of the values $y_v \in \{0,1\}$ to the vertices $v \in V$ in such a way that for each edge $e \in E$, the sum of the values on the endpoint vertices is at least $1$. All the vertices $v \in V$ with value $\tilde{y}_v=1$ are in the vertex cover. Thus, a \textit{vertex cover vector} of a graph $G(V,E)$ can be defined as a vector $\boldsymbol{y}=(y_v:v \in V)$ satisfying the following conditions:\begin{subequations}\label{vertexcov}
\begin{align}
&\sum_{\text{$v$ incident to $e$ }} y_v\geq 1, ~~~~~~~~ \text{for all} ~~ e \in E\\
&~~~~~y_v \in \{0,1\}, ~~~~~~~~~~~\text{for all} ~~ v \in V
\end{align}
\end{subequations}

\begin{definition}
A \ul{minimum vertex cover} in a graph is a vertex cover with minimum number of vertices.
\end{definition}
The cardinality of a minimum vertex cover in a graph $G(V,E)$ is called vertex cover number, denoted by $v(G)$. There may be several instances of a minimum vertex cover in a graph. Finding an instance of minimum vertex cover in a graph can be formulated as the following integer LP:
\begin{equation}
\label{minvertexcov}
\min~~{\sum_{v \in V}y_v}\qquad
\qquad\text{s.t.}~~\eqref{vertexcov}~~\text{holds}.
\end{equation}

\begin{proposition}\label{prop:duality}
For an arbitrary graph $G$, it is known that $m(G)\leq m_f(G)\leq v(G)$. For a bipartite graph $G$, it holds that $m(G)= m_f(G)= v(G)$.
\end{proposition}

In what follows, we assume that each server in the distributed storage system can serve up to one request at each moment, i.e., $\boldsymbol{\mu}=(\mu_1,\dots,\mu_n)=(1,\dots,1)$. Thus, $\mathcal{S}(\mathbf{G},\boldsymbol{\mu})$ and $\lambda^\star(\mathbf{G},\boldsymbol{\mu})$ only depend on the generator matrix $\mathbf{G}$ and are respectively denoted by $\mathcal{S}(\mathbf{G})$ and $\lambda^\star(\mathbf{G})$. Next, we present an example to show how the service rate of a code is connected to the matching and the vertex cover problems.

\subsection{Example of Equivalence}\label{example}
Here, we present an example to give more intuition about the subsequent results and to provide a sketch of the proofs. Consider a distributed storage system in which files $f_1$, $f_2$, and $f_3$ are stored across $7$ servers, labeled $1,\dots,7$, using a binary $[7,3]_2$ simplex code with the service rate $\mu_l=1$, $l \in [7]$. The generator matrix of this code is given by:

\[
     \mathbf{G}= \bordermatrix[{[]}]{%
      & 1 & 2 & 3 & 4 & 5 & 6 & 7 \cr
      & 1 & 0 & 1 & 0 & 1 & 0 & 1 \cr
      & 0 & 1 & 1 & 0 & 0 & 1 & 1 \cr
      & 0 & 0 & 0 & 1 & 1 & 1 & 1 \cr
     },
   \] 
   
where the number above each column shows the label of the corresponding column (server). Fig.~\ref{fig:DisStorSys} depicts this distributed storage system. 
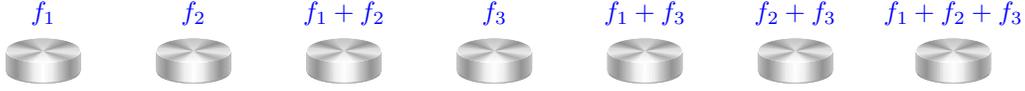
\begin{figure}[hbt]{
\begin{center}
\begin{tikzpicture}[scale=1]
\tikzset{pics/.cd,
  disc/.style={
    code={
      \fill [white] ellipse [x radius=2, y radius=2/3];
      \path [left color=black!50, right color=black!50, middle color=black!25] 
        (-2+.05,-1.1) arc (180:360:2-.05 and 2/3-.05*2/3) -- cycle;
      \path [top color=black!25, bottom color=white] 
        (0,.05*2/3) ellipse [x radius=2-.05, y radius=2/3-.05*2/3];
      \path [left color=black!25, right color=black!25, middle color=white] 
        (-2,0) -- (-2,-1) arc (180:360:2 and 2/3) -- (2,0) arc (360:180:2 and 2/3);
      \foreach \r in {225,315}
        \foreach \i [evaluate={\s=30;}] in {0,2,...,30}
          \fill [black, fill opacity=1/50] 
            (0,0) -- (\r+\s-\i:2 and 2/3) -- ++(0,-1) 
            arc (\r+\s-\i:\r-\s+\i:2 and 2/3) -- ++(0,1) -- cycle;
      \foreach \r in {45,135}
        \foreach \i [evaluate={\s=30;}] in {0,2,...,30}
          \fill [black, fill opacity=1/50] 
            (0,0) -- (\r+\s-\i:2 and 2/3) 
            arc (\r+\s-\i:\r-\s+\i:2 and 2/3)  -- cycle;
    }
  },
  disc bottom/.style={
    code={
      \foreach \i in {0,2,...,30}
        \fill [black, fill opacity=1/60] (0,-1.1) ellipse [x radius=2+\i/40, y radius=2/3+\i/60];
      \path pic[scale=0.2] {disc};
    }
  }
  }
\draw node  at (6.1,-2.5) {\small \textcolor{blue}{$f_1+f_2+f_3$}};
\path (6.1,-3.) pic[scale=0.05] {disc bottom} (6.1,-3.) pic[scale=0.25] {disc};
\draw node  at (4,-2.5) {\small \textcolor{blue}{$f_2+f_3$}};
\path (4,-3.) pic[scale=0.05] {disc bottom} (4,-3.) pic[scale=0.25] {disc};
\draw node  at (2,-2.5) {\small \textcolor{blue}{$f_1+f_3$}};
\path (2,-3.) pic[scale=0.05] {disc bottom} (2,-3.) pic[scale=0.25] {disc};
\draw node  at (0,-2.5) {\small \textcolor{blue}{$f_3$}};
\path (0,-3.) pic[scale=0.05] {disc bottom} (0,-3.) pic[scale=0.25] {disc};
\draw node  at (-2,-2.5) {\small \textcolor{blue}{$f_1+f_2$}};
\path (-2,-3.) pic[scale=0.05] {disc bottom} (-2,-3.) pic[scale=0.25] {disc};
\draw node  at (-4,-2.5) {\small \textcolor{blue}{$f_2$}};
\path (-4,-3.) pic[scale=0.05] {disc bottom} (-4,-3.) pic[scale=0.25] {disc};
\draw node  at (-6,-2.5) {\small \textcolor{blue}{$f_1$}};
\path (-6,-3.) pic[scale=0.05] {disc bottom} (-6,-3.) pic[scale=0.25] {disc};
\end{tikzpicture}\vspace{-0.5cm}
\end{center}
}
\caption{A distributed storage system consists of $7$ servers storing files $f_1$, $f_2$, and $f_3$ using a binary $[7,3]_2$ simplex code.}\label{fig:DisStorSys}
\end{figure}

The recovery sets for each file are given by
\begin{align*}
& \mathcal{R}_{1}=\{R_{1,1},\dots,R_{1,4}\}=\{\{0,1\},\{2,3\},\{4,5\},\{6,7\}\}\\
& \mathcal{R}_{2}=\{R_{2,1},\dots,R_{2,4}\}=\{\{0,2\},\{1,3\},\{4,6\},\{5,7\}\}   
\\ &
\mathcal{R}_{3}=\{R_{3,1},\dots,R_{3,4}\}=\{\{0,4\},\{1,5\},\{2,6\},\{3,7\}\}  
\end{align*}

The graph representation of $[7,3]_2$ simplex code is drawn in Fig.~\ref{fig:GraphRep}, which is a bipartite graph. The vertices ${\emptyset_{f_1}}$, ${\emptyset_{f_2}}$ and ${\emptyset_{f_3}}$ are the dummy vertices added to the graph for the purpose of removing the self-loops of systematic vertices $f_1$, $f_2$, and $f_3$, respectively. The edges with color magenta, green, and blue represent recovery sets for files $f_1$, $f_2$, and $f_3$, respectively. Moreover, the label $\lambda_{i,j}$ above an edge indicates the portion of requests for file $f_i$ that is assigned to the recovery set $R_{i,j}$. 
\begin{figure}[hbt]
\begin{center}
\begin{tikzpicture}[scale=1]
\def\horzgap{2.5in}; 
\def \gapVN{0.6in}; 
\def \gapCN{0.6in}; 
\def\nodewidth{0.4in};
\def \edgewidth{0.05in};

\tikzstyle{right} = [circle, draw, line width=0.05mm, inner sep=0mm, fill=white, minimum size=\nodewidth]
\tikzstyle{left} = [circle, draw, line width=0.05mm, inner sep=0mm, fill=white, minimum size=\nodewidth]

\foreach \vn in {1,2,3,4}{
 \node[left] (vn\vn) at (0,-\vn*\gapVN) {};
}                    

\path (vn1) ++(0,0) node()[scale=0.3, inner sep=0mm] {\Huge{$\color{black}{f_1}$}};
\path (vn2) ++(0,0) node()[scale=0.3, inner sep=0mm] {\Huge{$\color{black}{f_2}$}};
\path (vn3) ++(0,0) node()[scale=0.3, inner sep=0mm] {\Huge{$\color{black}{f_3}$}};
\path (vn4) ++(0,0) node()[scale=0.25, inner sep=0mm] {\huge{$\color{black}{f_1+f_2+f_3}$}};

\path (vn1)++(-0.8cm,0in) node ()[scale=0.33] {\Huge{$\color{black}{1}$}};
\path (vn2)++(-0.8cm,0in) node ()[scale=0.33] {\Huge{$\color{black}{2}$}};
\path (vn3)++(-0.8cm,0in) node ()[scale=0.33] {\Huge{$\color{black}{4}$}};
\path (vn4)++(-0.8cm,0in) node ()[scale=0.33] {\Huge{$\color{black}{7}$}};

\foreach \cn in {1,...,6}{
\node[right] (cn\cn) at (\horzgap,-\cn*\gapCN+0.5in) {};
}

\draw[line width=0.25mm, magenta] (vn1.east)--(cn1.west) node [midway,above] {$\lambda_{1,1}$};
\draw[line width=0.25mm, green] (vn2.east)--(cn2.west)node [midway,above] {$\lambda_{2,1}$};
\draw[line width=0.25mm, blue] (vn3.east)--(cn3.west) node [pos=0.9,above] {$\lambda_{3,1}$};

\draw[line width=0.25mm, green] (vn1.east)--(cn4.west)node [pos=0.6,above] {$\lambda_{2,2}$};
\draw[line width=0.25mm, magenta] (vn2.east)--(cn4.west)node [pos=0.25,above] {$\lambda_{1,2}$};
\draw[line width=0.25mm, blue] (vn4.east)--(cn4.west)node [pos=0.6,above] {$\lambda_{3,4}$};

\draw[line width=0.25mm, blue] (vn1.east)--(cn5.west)node [pos=0.9,above] {$\lambda_{3,2}$};
\draw[line width=0.25mm, magenta] (vn3.east)--(cn5.west)node [pos=0.25,above] {$\lambda_{1,3}$};
\draw[line width=0.25mm, green] (vn4.east)--(cn5.west)node [pos=0.7,above] {$\lambda_{2,4}$};
              
\draw[line width=0.25mm, blue] (vn2.east)--(cn6.west)node [pos=0.9,above] {$\lambda_{3,3}$};
\draw[line width=0.25mm, green] (vn3.east)--(cn6.west)node [pos=0.35,above] {$\lambda_{2,3}$};
\draw[line width=0.25mm, magenta] (vn4.east)--(cn6.west)node [midway,above] {$\lambda_{1,4}$};

\def\moveX {1.8*\nodewidth};
\def\moveXA {2*\nodewidth};

\path (cn1)++(0,0in) node ()[scale=0.3] {\Huge{$\color{black}{\emptyset_{f_1}}$}};
\path (cn2)++(0,0in) node ()[scale=0.3] {\Huge{$\color{black}{\emptyset_{f_2}}$}};
\path (cn3)++(0,0in) node ()[scale=0.3] {\Huge{$\color{black}{\emptyset_{f_3}}$}};
\path (cn4)++(0,0in) node ()[scale=0.3] {\Huge{$\color{black}{f_1+f_2}$}};
\path (cn5)++(0,0in) node ()[scale=0.3] {\Huge{$\color{black}{f_1+f_3}$}};
\path (cn6)++(0,0in) node ()[scale=0.3] {\Huge{$\color{black}{f_2+f_3}$}};

\path (cn1)++(0.8cm,0cm) node ()[scale=0.33] {\Huge{$\color{black}{0}$}};
\path (cn2)++(0.8cm,0cm) node ()[scale=0.33] {\Huge{$\color{black}{0}$}};
\path (cn3)++(0.8cm,0cm) node ()[scale=0.33] {\Huge{$\color{black}{0}$}};
\path (cn4)++(0.8cm,0cm) node ()[scale=0.33] {\Huge{$\color{black}{3}$}};
\path (cn5)++(0.8cm,0cm) node ()[scale=0.33] {\Huge{$\color{black}{5}$}};
\path (cn6)++(0.8cm,0cm) node ()[scale=0.33] {\Huge{$\color{black}{6}$}};

\end{tikzpicture}\vspace{-0.5cm}
\end{center}
\caption{Graph representation of $[7,3]_2$ simplex code.}\label{fig:GraphRep}
\end{figure}
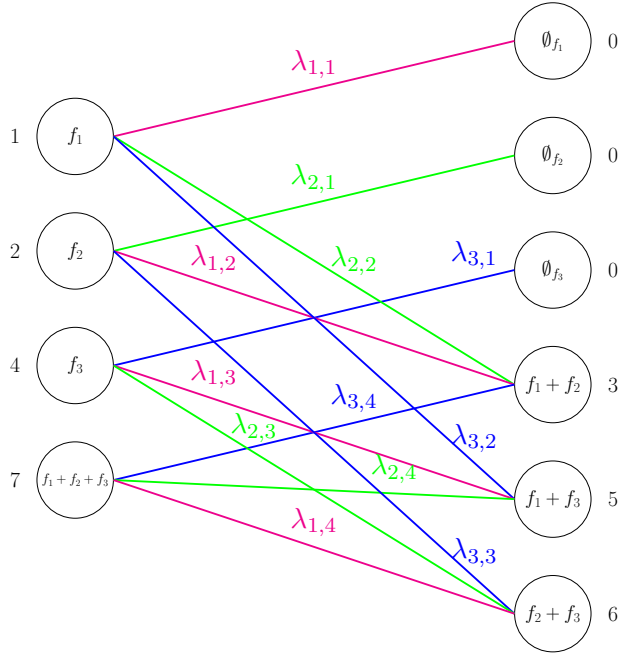

The service rate region $\mathcal{S}(\mathbf{G})$ of this system is the set of vectors ${\boldsymbol{\lambda}}=(\lambda_{1},\lambda_{2},\lambda_{3})$ for which there exist $\lambda_{i,j}$'s, $i \in [3]$ and $j \in [4]$, satisfying the set of constraints~\eqref{eq:1} as follows:
\begin{align}
&\text{\eqref{eq:1a}}\Rightarrow
\begin{cases}
\lambda_{1}=\lambda_{1,1}+\lambda_{1,2}+\lambda_{1,3}+\lambda_{1,4}\\
\lambda_{2}=\lambda_{2,1}+\lambda_{2,2}+\lambda_{2,3}+\lambda_{2,4}\\
\lambda_{3}=\lambda_{3,1}+\lambda_{3,2}+\lambda_{3,3}+\lambda_{3,4}
\end{cases}\label{eq:ExEqual}\\
&\text{\eqref{eq:1b}}\Rightarrow
\begin{cases}
\lambda_{1,1}+\lambda_{2,2}+\lambda_{3,2} \leq 1\\
\lambda_{2,1}+\lambda_{1,2}+\lambda_{3,3} \leq 1\\
\lambda_{3,1}+\lambda_{1,3}+\lambda_{2,3} \leq 1\\
\lambda_{3,4}+\lambda_{2,4}+\lambda_{1,4} \leq 1\\
\lambda_{2,2}+\lambda_{1,2}+\lambda_{3,4} \leq 1\\
\lambda_{3,2}+\lambda_{1,3}+\lambda_{2,4} \leq 1\\
\lambda_{3,3}+\lambda_{2,3}+\lambda_{1,4} \leq 1
\end{cases}\label{eq:ExInequal}\\
&\text{\eqref{eq:pos}}\Rightarrow
\begin{cases}
\lambda_{i,j} \in \mathbb{R}_{\geq 0},~~~~ \text{for all} ~~~ i\in [3],~j\in [4]
\end{cases}\label{eq:ExBound}
\end{align}\vspace{0.01cm}

Fig.~\ref{SerRateReg} shows the service rate region $\mathcal{S}(\mathbf{G})$ of this coded storage system. system.\vspace{-0.1cm}
\begin{figure}[hbt]
\begin{center}
\begin{tikzpicture}[>=stealth,tdplot_main_coords]
    \coordinate (O) at (0,0,0);
    \draw[->] (O) -- (3.5,0,0) coordinate (x axis); 
        \foreach \x/\xtext in {2.5} 
           \draw (\x, 0.1,0) -- (\x, -0.1,0) node [below] {$4$};
    \draw[->] (O) -- (0,3.5,0) coordinate (y axis);
        \foreach \y/\ytext in {2.5}
           \draw (0.1,\y,0) -- (-0.1,\y,0) node [below] {$4$};
    \draw[->] (O) -- (0,0,2.75) coordinate (z axis);
        \foreach \z/\ztext in {2}
            \draw (-0.05,0.05,\z) -- (0.05,-0.05,\z) node [left] {$4$};

    \coordinate[label=above:{\large $\lambda_1$}] (lambda_a) at (3.5, 0, 0.1);
    \coordinate[label=above:{\large $\lambda_2$}] (lambda_b) at (0, 3.5, 0.1);
    \coordinate[label=right:{\large $\lambda_3$}] (lambda_c) at (0, 0.1, 2.75);

    \filldraw[draw=blue,fill=blue,fill opacity=0.2] (0,0,2) -- (2.5,0,0) -- (0,2.5,0) -- (0,0,2) -- (0,0,0) -- (0,2.5,0);
    \filldraw[draw=blue,fill=blue,fill opacity=0.2] (0,0,0) -- (2.5,0,0);
    
\end{tikzpicture}\vspace{-0.5cm}
\end{center}
\caption{Service rate region of $[7,3]_2$ simplex code.}\label{SerRateReg}
\end{figure}
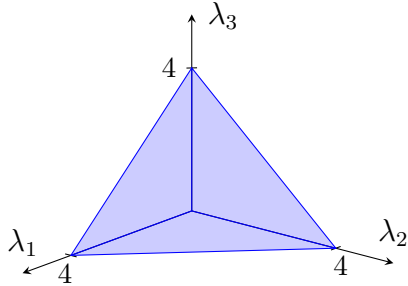

Based on~\eqref{FracMatch}, a fractional matching vector ${\boldsymbol{x}=(\lambda_{1,1},\dots,\lambda_{1,4},\lambda_{2,1},\dots,\lambda_{2,4},\lambda_{3,1},\dots,\lambda_{3,4})}$ of the graph depicted in Fig.~\ref{fig:GraphRep}, satisfies the constraints~\eqref{eq:ExInequal} and \eqref{eq:ExBound}. Thus, according to Definition~\ref{def:SerRateReg}, a vector ${{\boldsymbol{\lambda}}=(\lambda_{1},\lambda_{2},\lambda_{3})}$ obtained from $\boldsymbol{x}$ using~\eqref{eq:ExEqual} is in the service rate region of $[7,3]_2$ simplex code. Conversely, for a vector ${\boldsymbol{\lambda}}$ in the service rate region of $[7,3]_2$ simplex code, there exist ${\lambda_{i,j}}$'s, ${i\in [3]}$ and ${j\in [4]}$, satisfying the constraints~\eqref{eq:ExInequal} and \eqref{eq:ExBound}, that define a fractional matching vector ${\boldsymbol{x}=(\lambda_{i,j}:i\in [3]~\text{and}~j\in [4])}$ in the graph representation of $[7,3]_2$ simplex code drawn in Fig.~\ref{fig:GraphRep}.

Based on~\eqref{MaxFracMatch}, a maximum fractional matching vector $\boldsymbol{x}^\star$ is obtained by solving the following LP:
\begin{align}\label{exmaxmatch}
\max~~{\sum_{i=1}^{3}\sum_{j=1}^{4}\lambda_{i,j}}\qquad
\text{s.t.}~~\eqref{eq:ExInequal}~\text{and}~\eqref{eq:ExBound}
~~\text{hold}.
\end{align}
We want to show that the vector ${{\boldsymbol{\lambda}}=(\lambda_{1},\lambda_{2},\lambda_{3})}$ obtained from $\boldsymbol{x}^\star$ using~\eqref{eq:ExEqual} is in fact a maximum demand vector ${\boldsymbol{\lambda}}^\star$ in the service rate region of $[7,3]_2$ simplex code. From \eqref{eq:ExEqual}, ${\sum_{i=1}^{3}\sum_{j=1}^{4}\lambda_{i,j}={\lambda_{1}+\lambda_{2}+\lambda_{3}}}$. Thus, it can be easily verified that $\boldsymbol{x}^\star$ provides a solution for the following LP:
\begin{align}\label{exmaxrate}
\max~~{\lambda_{1}+\lambda_{2}+\lambda_{3}}\qquad
\text{s.t.}~~\eqref{eq:ExEqual},~\eqref{eq:ExInequal},~\eqref{eq:ExBound}~~\text{hold}.
\end{align}      
Moreover, according to~\eqref{maxrate}, an instance of maximum demand vector is obtained by solving the LP in~\eqref{exmaxrate}. Thus, the vector ${{\boldsymbol{\lambda}}=(\lambda_{1},\lambda_{2},\lambda_{3})}$ obtained from $\boldsymbol{x}^\star$ using~\eqref{eq:ExEqual} is a 
maximum demand vector ${\boldsymbol{\lambda}}^\star$. On the other hand, for an instance of  $\boldsymbol{{\lambda}}^\star$ in the service rate region of $[7,3]_2$ simplex code obtained from~\eqref{exmaxrate}, there exists a fractional matching vector $\boldsymbol{x}$ which according to the same reasoning, provides a solution for~\eqref{exmaxmatch}. Thus, the vector $\boldsymbol{x}$ is a maximum fractional matching vector $\boldsymbol{x}^\star$ in the graph representation of $[7,3]_2$ simplex code in Fig.~\ref{fig:GraphRep}. Since a maximum demand vector ${\boldsymbol{{\lambda}}^\star=(\lambda_{1}^\star,\lambda_{2}^\star,\lambda_{3}^\star)}$ is obtained from a maximum fractional matching vector $\boldsymbol{x}^\star$ by~\eqref{eq:ExEqual}, it follows that ${\lambda_{1}^\star+\lambda_{2}^\star+\lambda_{3}^\star=\sum\lambda_{i,j}^\star}$, where $\lambda_{i,j}^\star$'s are the elements of ${\boldsymbol{x}}^\star$. Hence, we have ${\lambda^\star(\mathbf{G})=m_f(G)}$, and based on Proposition~\ref{prop:duality}, ${{m(G)\leq \lambda^\star(\mathbf{G})\leq v(G)}}$ holds.

We show that the service capacity of $[7,3]_2$ simplex code is $4$. The proof consists of two parts. First, we need to prove the converse by showing that the service capacity cannot be bigger than $4$. It is easy to see that the set of vertices $\{f_1,f_2,f_3,f_1+f_2+f_3\}$ is a minimum vertex cover for the graph in Fig.~\ref{fig:GraphRep}. Thus, the vertex cover number of this graph is $v(G)=4$ which indicates that $\lambda^\star(\mathbf{G})\leq 4$. Next, we show the achievability proof by showing that there exists a demand vector ${{\boldsymbol{\lambda}}=(\lambda_{1},\lambda_{2},\lambda_{3})}$ in the service rate region such that $\lambda_1+\lambda_2+\lambda_3=4$. For this purpose, one can consider the set of edges labeled by $\lambda_{1,1}$, $\lambda_{1,2}$, $\lambda_{1,3}$, and $\lambda_{1,4}$ as a matching in the graph. Corresponding to this matching, a demand vector ${{\boldsymbol{\lambda}}=(4,0,0)}$ is obtained by applying~\eqref{eq:ExEqual}.

\subsection{Equivalence Results}\label{Equivalence-Results}
We first show an equivalence between the service rate problem and the fractional matching problem. This equivalence result allow us to derive bounds on the service capacity of a coded storage system and then to recover the service capacity of the binary simplex code whose graph representation is bipartite. 

\begin{theorem}\label{theorem1}
Consider an $(\mathbf{G},\boldsymbol{\mu})$ system with the service rate $\boldsymbol{\mu}=\mathbf{1}_n$. There exists a demand vector $\boldsymbol{\lambda}=(\lambda_1,\dots,\lambda_k)$ in the service rate region of this system if and only if there exists a fractional matching vector ${\boldsymbol{x}=(\lambda_{i,j}:i\in [k]~\text{and}~j\in [t_i])}$ in the graph representation of $[n,k]_q$ code such that $\boldsymbol{\lambda}$ and $\boldsymbol{x}$ are related based on~\eqref{eq:1a}.
\end{theorem}

\begin{Corollary}\label{corollary1}
Consider an ${(\mathbf{G},\boldsymbol{\mu})}$ system with ${\boldsymbol{\mu}=\mathbf{1}_n}$. There exists a maximum demand vector $\boldsymbol{\lambda}^\star=(\lambda_1^\star,\dots,\lambda_k^\star)$ in the service rate region $S(\mathbf{G})$ of this storage system if and only if there exists a maximum fractional matching vector ${\boldsymbol{x}^\star=(\lambda_{i,j}^\star:i\in [k]~\text{and}~j\in [t_i])}$ in the graph representation of $[n,k]_q$ code such that $\boldsymbol{\lambda}^\star$ and $\boldsymbol{x}^\star$ are related based on~\eqref{eq:1a}.
\end{Corollary}

\begin{theorem}\label{theorem2}
Consider an ${(\mathbf{G},\boldsymbol{\mu})}$ system with the service rate ${\boldsymbol{\mu}=\mathbf{1}_n}$. The service capacity ${\lambda^\star(\mathbf{G})}$ of this system is lower bounded by the matching number and upper bounded by the vertex cover number of the graph representation of $[n,k]_q$ code. i.e., ${m(G) \leq \lambda^\star(\mathbf{G})=m_f(G) \leq v(G)}$.
\end{theorem}

Note that if the graph representation of a code is bipartite, Proposition~\ref{prop:duality} results ${m(G)=\lambda^\star(\mathbf{G})=v(G)}$.

\begin{theorem}\label{theorem3}
The graph representation of $[2^k-1,k,2^{k-1}]_2$ simplex code, is a bipartite graph.
\end{theorem}

\begin{Corollary}\label{cor:simplex}
For an $(\mathbf{G},\boldsymbol{\mu})$ system with  ${[2^k-1,k,2^{k-1}]_2}$ simplex code and service rate ${\boldsymbol{\mu}=\mathbf{1}_n}$, the service capacity is given by ${m(G)=\lambda^\star(\mathbf{G})=v(G)=2^{k-1}}$. 
\end{Corollary}

\begin{Corollary}\label{cor:integraleq}
Consider an $(\mathbf{G},\boldsymbol{\mu})$ system with ${\boldsymbol{\mu}=\mathbf{1}_n}$. There exists a demand vector $\boldsymbol{\lambda}=(\lambda_1,\dots,\lambda_k)$ in the integral service rate region $S_I(\mathbf{G})$ of this system if and only if there exists a matching vector ${\tilde{\boldsymbol{x}}=(\lambda_{i,j}:i\in [k]~\text{and}~j\in [t_i])}$ in the graph representation of $[n,k]_q$ code such that $\boldsymbol{\lambda}$ and $\tilde{\boldsymbol{x}}$ are related based on~\eqref{eq:1a}.
\end{Corollary}

\begin{Corollary}\label{cor:integralmaxeq}
Consider an $(\mathbf{G},\boldsymbol{\mu})$ system with ${\boldsymbol{\mu}=\mathbf{1}_n}$. There exists a maximum demand vector $\boldsymbol{\lambda}^\star=(\lambda_1^\star,\dots,\lambda_k^\star)$ in the integral service rate region $S_I(\mathbf{G})$ of this storage system if and only if there exists a maximum matching vector ${\tilde{\boldsymbol{x}}^\star=(\lambda_{i,j}^\star:i\in [k]~\text{and}~j\in [t_i])}$ in the graph representation of $[n,k]_q$ code such that $\boldsymbol{\lambda}^\star$ and $\tilde{\boldsymbol{x}}^\star$ are related based on~\eqref{eq:1a}.
\end{Corollary}

\section{Generalization of Batch codes}\label{generalization-batch}

In this section, we show how the service rate problem can be viewed as a generalization of the problem of batch codes. That further illustrates connections with PIR codes, switch codes and locally repairable codes which all can be seen as special cases of batch codes (see~\cite{skachek2018batch}). 

\subsection{Definitions of Batch Codes and PIR Codes}

\begin{definition}\cite{ishai2004batch}
An \ul{$(n, k, t, m, \tau)$ batch code} $\mathcal{C}$ over a finite alphabet ${\sum}$ encodes any string $\boldsymbol{x}=(x_1,\dots,x_k)$ into $m$ strings (buckets) $\boldsymbol{y}_1,\dots,\boldsymbol{y}_m$ of total length $n$ by an encoding mapping $\mathcal{C}: \sum^k \rightarrow \sum^n$, such that for each $t$-tuple (batch) of indices $i_1,\dots,i_t \in [k]$, the entries $x_{i_1},\dots,x_{i_t}$ can be decoded by reading at most $\tau$ symbols from each bucket. 
\end{definition}

\begin{definition}\cite{skachek2018batch}
An \ul{$(n, k, t)$ primitive batch code} is an $(n, k, t, m, \tau)$ batch code, where each bucket contains exactly one symbol, i.e., $n=m$. Note that in this setting $\tau=1$, i.e., at most one symbol can be recovered from each bucket.
\end{definition}

\begin{definition}
An \ul{$(n, k, t)$ multiset primitive batch code} is an $(n, k, t)$ primitive batch code where the information symbols ${x_{i_1},\dots, x_{i_t}}$ are requested by $t$ distinct users such that the indices ${i_1,\dots,i_t}$ are not necessarily distinct and in general they form a multiset. Moreover, the requested symbols can be reconstructed from the data read by $t$ different users independently (i.e., $x_{i_j}$ can be recovered by the user $j$) so that the sets of the symbols read by these users are disjoint.
\end{definition}

It should be noted that for the sake of simplicity, we refer to a linear $(n, k, t)$ multiset primitive batch code over $\mathbb{F}_q$ as an \ul{$[n,k,t]_q$ batch code}.

\begin{proposition}\cite[Theorem 1]{lipmaa2015linear} \label{batchcode}
A linear $[n,k]_q$ code $\mathcal{C}$ with generator matrix $\mathbf{G}$ is an $[n,k,t]_q$ batch code if and only if there exist $t$ non-intersecting sets $T_1,\dots,T_t$ of indices of columns in the generator matrix $\mathbf{G}$ such that for each $j \in [t]$, there exists a linear combination of columns of $\mathbf{G}$ indexed by $T_j$ which equals to the vector $\mathbf{e}_{i_j}$, for all $j \in [t]$ and $i_j \in [k]$. 
\end{proposition}

\begin{definition}\cite{fazeli2015pir}\label{PIR}
A linear $[n,k]_2$ code $\mathcal{C}$ with generator matrix $\mathbf{G}$ is called a t-server PIR code if for every $i \in [k]$, there exist $t$ disjoint sets of columns of $\mathbf{G}$ that add up to $\mathbf{e}_i$.
\end{definition}

\subsection{Connection with Batch Codes and PIR Codes}

\begin{theorem}\label{theorem4}
Given the integral service rate region $\mathcal{S}_I(G)$ of code $\mathbf{G} \in \mathbb{F}^{k\times n}_q$ with service rate ${\boldsymbol{\mu}=\mathbf{1}_n}$, if all vectors in the set $S_t=\{\boldsymbol{\lambda}=(\lambda_1,\dots,\lambda_k)|\sum_{i=1}^{k}\lambda_i=t, \lambda_i \in \mathbb{Z}_{\geq 0}\}$ are in the $\mathcal{S}_I(G)$, the code $\mathbf{G}$ is a linear $[n,k,t]_q$ batch code.
\end{theorem}

Theorem~\ref{theorem4} shows that the integral setting of the service rate problem where the solution (the portion of requests that are assigned to the recovery sets) is restricted to be integral, is the same as the setting of the multiset primitive batch code problem. Thus, the general setting of the service rate problem where a fractional solution is allowed, can be viewed as a generalization of the setting of the multiset primitive batch code problem.

\begin{Corollary}\label{cor:genPIR}
Given the integral service rate region $\mathcal{S}_I(G)$ of code $\mathbf{G} \in \mathbb{F}^{k\times n}_q$ with service rate ${\boldsymbol{\mu}=\mathbf{1}_n}$, if all vectors in the set $S_t=\{t.e_1=(t,0,\dots,0), \dots, t.e_k=(0,\dots,0,t)| t \in \mathbb{N}\}$ are in the $\mathcal{S}_I(G)$, the code $\mathbf{G}$ is a t-server PIR code.
\end{Corollary}

Next, we present an example regarding the application of Theorems~\ref{theorem4} that shows a binary $[7,3]_2$ simplex code is a $[7,3,4]_2$ batch code. 

\begin{example}
Consider a binary $[7,3]_2$ simplex code. In this example, utilizing the graph representation and the integral service rate region of the code, we want to show that this code is a $[7,3,4]_2$ batch code. To this end, we need to show that each demand vector $\boldsymbol{\lambda}=(\lambda_1,\lambda_2,\lambda_3)$ with $\sum_{i=1}^{3}\lambda_i=4$, is in the integral service rate region of the $[7,3]_2$ simplex code, i.e., for each of these vectors, there exists a matching in the graph representation of $[7,3]_2$ binary simplex code. W.l.o.g we assume that $\lambda_1\geq \lambda_2 \geq \lambda_3$. The $4$ recovery sets of each file, say $f_1$, are known. As can be seen in Fig.~\ref{fig:lattice}, the four magenta edges corresponding to the recovery sets of file $f_1$, constructs a maximum matching. Using the Algorithm~\ref{algorithm}, we show that how one can start with a maximum matching corresponding to the vector $\boldsymbol{\lambda_a}=(4,0,0)$ and by following at most two steps find the maximum matching corresponding to any  vector ${\boldsymbol{\lambda_b}=(\lambda_1,\lambda_2,\lambda_3)}$ with ${\lambda_1+\lambda_2+\lambda_3=4}$. For this purpose, in a nutshell, we start with the recovery sets of file $f_1$ and replace some of them with the recovery sets for files $f_2$ and $f_3$ as needed. Next, we define three steps, based on which we present the Algorithm~\ref{algorithm} that can be generalized for any number of files $k$.\vspace{0.1cm}

\textbf{Step 1:} Consider the systematic recovery set for file $f_i$, and add the corresponding edge to the matching set. Accordingly, remove the recovery set for file $f_1$, incident to the node $f_i$, from the matching set.

\textbf{Step 2:} Find ${{(\lambda_i-1)}/{2}}$ number of loops, each of size $4$, consisting of $2$ recovery sets for file $f_i$ and $2$ recovery sets for file $f_1$, in the graph representation of the code. Then, by considering each of the loops, replace the $2$ recovery sets for file $f_1$ with the $2$ recovery sets for file $f_i$ in the matching. 

\textbf{Step 3:} This step would be the same as step $2$, except that ${(\lambda_i-1)/2}$ is replaced with $(\lambda_i)/2$.

\begin{algorithm}
 \textbf{Input: }Max matching corresponding to $(4,0,0)$ \\
 \textbf{for} $i=2:3$ \textbf{do}\\
  ~~\textbf{if} $\lambda_i$ is odd\\
  ~~~~~~\textbf{do} Step 1;\\
  ~~~~~~\textbf{do} Step 2;\\
  ~~\textbf{else}\\
  ~~~~~~\textbf{do} Step 3;\\
 \textbf{end}\\
\textbf{Output: }Max matching corresponding to $(\lambda_1,\lambda_2,\lambda_3)$\\
 \caption{Finding a Max Matching for any $(\lambda_1,\lambda_2,\lambda_3)$ with ${\lambda_1+\lambda_2+\lambda_3=4}$}  \label{algorithm}
\end{algorithm}

For instance, for finding the corresponding matching for the demand vector ${\boldsymbol{\lambda}=(2,2,0)}$, we need to show that there are two recovery sets of file $f_1$ that can be used to form two recovery sets for file $f_2$. The $4$ magenta edges (recovery sets of file $f_1$) form a maximum matching of size $4$. In the graph representation, it is easy to find a loop of size $4$ consisting of two magenta edges, $\lambda_{a,3}$ and $\lambda_{a,4}$, and two green edges (recovery sets of file $f_2$), $\lambda_{b,3}$ and $\lambda_{b,4}$. Therefore, simply we can replace these two magenta edges with the green edges and construct another maximum matching of size $4$ which is a matching corresponding to the demand ${\boldsymbol{\lambda}=(2,2,0)}$. 

Now, consider the demand vector ${\boldsymbol{\lambda}=(2,1,1)}$. Since $\lambda_2$ and $\lambda_3$ are odd, we need to find the recovery sets of file $f_1$ that can be used for building systematic recovery sets for files $f_2$ and $f_3$. It can be seen that the magenta edges connected to the nodes $2$ and $4$ of the graph representation, i.e., $\lambda_{a,2}$ and $\lambda_{a,3}$, can be removed from the original maximum matching and be substituted by the green edge $\lambda_{b,1}$, and the blue edge $\lambda_{c,1}$, which represent systematic recovery sets for $f_1$ and $f_2$, respectively.
\end{example}

\begin{figure}[hbt]
\begin{center}
\begin{tikzpicture}
\def\horzgap{0.9in}; 
\def \gapVN{1in}; 
\def \gapCN{1in}; 
\def\nodewidth{0.4in};
\def \edgewidth{0.05in};

\tikzstyle{righ} = [circle, draw, line width=0.05mm, inner sep=0mm, fill=white, minimum size=\nodewidth]
\tikzstyle{lef} = [circle, draw, line width=0.05mm, inner sep=0mm, fill=white, minimum size=\nodewidth]
\tikzstyle{end} = [circle, draw, line width=0.05mm, inner sep=0mm, fill=white, minimum size=1.3*\nodewidth]

\foreach \a in {1,...,3}{ \node[righ] (a\a) at (1.3*\a*\gapCN+0in,\horzgap) {};}

\path (a1)++(0,0in) node ()[scale=0.3] {\Huge{$\color{black}{\emptyset_{f_1}}$}};
\path (a1)++(-0.7cm,0in) node ()[scale=0.3] {\Huge{$\color{black}{0}$}};
\path (a2)++(0,0in) node ()[scale=0.3] {\Huge{$\color{black}{\emptyset_{f_2}}$}};
\path (a2)++(-0.7cm,0in) node ()[scale=0.3] {\Huge{$\color{black}{0}$}};
\path (a3)++(0,0in) node ()[scale=0.3] {\Huge{$\color{black}{\emptyset_{f_3}}$}};
\path (a3)++(-0.7cm,0in) node ()[scale=0.3] {\Huge{$\color{black}{0}$}};

\foreach \b in {1,...,3}{
 \node[lef] (b\b) at (1.3*\b*\gapVN+0in,0) {};
}    
\path (b1) ++(0,0) node()[scale=0.25, inner sep=0mm] {\Huge{$\color{black}{f_1}$}};
\path (b1) ++(-0.7cm,0) node()[scale=0.25, inner sep=0mm] {\Huge{$\color{black}{1}$}};
\path (b2) ++(0,0) node()[scale=0.25, inner sep=0mm] {\Huge{$\color{black}{f_2}$}};
\path (b2) ++(-0.7cm,0) node()[scale=0.25, inner sep=0mm] {\Huge{$\color{black}{2}$}};
\path (b3) ++(0,0) node()[scale=0.25, inner sep=0mm] {\Huge{$\color{black}{f_3}$}};
\path (b3) ++(-0.7cm,0) node()[scale=0.25, inner sep=0mm] {\Huge{$\color{black}{4}$}};

\foreach \c in {1,...,3}{
 \node[lef] (c\c) at (0.5*\c*+2.6in,-1.2*\horzgap) {};
} 
\path (c1)++(0,0in) node ()[scale=0.3] {\Huge{$\color{black}{f_1+f_2}$}};
\path (c1)++(-0.7cm,0in) node ()[scale=0.3] {\Huge{$\color{black}{3}$}};
\path (c2)++(0,0in) node ()[scale=0.3] {\Huge{$\color{black}{f_1+f_3}$}};
\path (c2)++(-0.7cm,0in) node ()[scale=0.3] {\Huge{$\color{black}{5}$}};
\path (c3)++(0,0in) node ()[scale=0.3] {\Huge{$\color{black}{f_2+f_3}$}};
\path (c3)++(-0.7cm,0in) node ()[scale=0.3] {\Huge{$\color{black}{6}$}};

\foreach \d in {1}{
 \node[end] (d\d) at (1*\d+2.6in,-2.5*\horzgap) {};
}   
\path (d1) ++(0,0) node()[scale=0.25, inner sep=0mm] {\Huge{$\color{black}{f_1+f_2+f_3}$}};
\path (d1) ++(-0.8cm,0) node()[scale=0.25, inner sep=0mm] {\Huge{$\color{black}{7}$}};

\draw[line width=0.25mm, magenta] (a1.south)--(b1.north) node [midway,left] {$\lambda_{a,1}$};
\draw[line width=0.25mm, magenta] (b2.south)--(c1.north)node [pos=0.6,below] {$\lambda_{a,2}$};
\draw[line width=0.25mm, magenta] (b3.south)--(c2.north)node [pos=0.6,below] {$\lambda_{a,3}$};
\draw[line width=0.25mm, magenta] (c3.south)--(d1.north)node [midway,above] {$\lambda_{a,4}$};

\draw[line width=0.25mm, green] (a2.south)--(b2.north) node [midway,left] {$\lambda_{b,1}$};
\draw[line width=0.25mm, green] (b1.south)--(c1.north)node [midway,left] {$\lambda_{b,2}$};
\draw[line width=0.25mm, green] (b3.south)--(c3.north)node [midway,left] {$\lambda_{b,3}$};
\draw[line width=0.25mm, green] (c2.south)--(d1.north)node [midway,left] {$\lambda_{b,4}$};

\draw[line width=0.25mm, blue] (a3.south)--(b3.north) node [midway,left] {$\lambda_{c,1}$};
\draw[line width=0.25mm, blue] (b1.south)--(c2.north)node [pos=0.4,above] {$\lambda_{c,2}$};
\draw[line width=0.25mm, blue] (b2.south)--(c3.north)node [pos=0.4,above] {$\lambda_{c,3}$};
\draw[line width=0.25mm, blue] (c1.south)--(d1.north)node [midway,above] {$\lambda_{c,4}$};

\def\moveX {1.8*\nodewidth};
\def\moveXA {2*\nodewidth};

\end{tikzpicture} \vspace{-0.25cm}
\end{center}
\caption{Graph representation of $[7,3]_2$ simplex code.}\label{fig:lattice}
\end{figure}
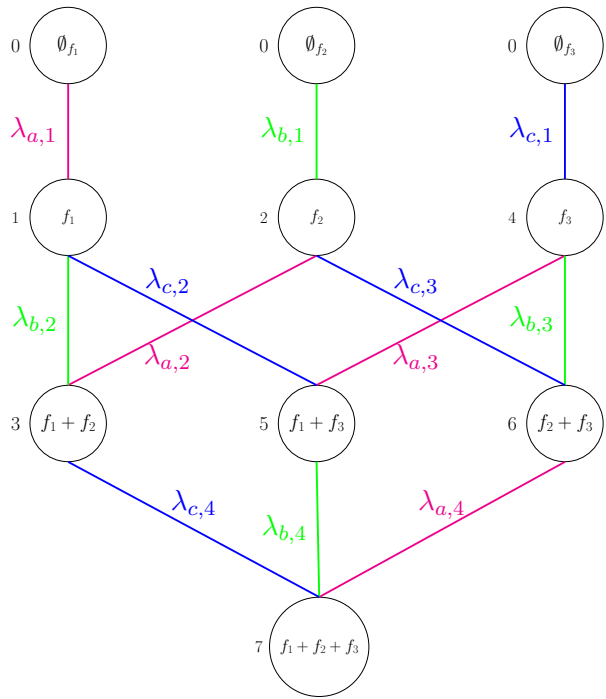

\section*{Acknowledgment}
Part of this research is based upon work supported by the National Science Foundation under Grant No.~CIF-1717314.

\bibliographystyle{IEEEtran}
\bibliography{ISIT-2020-Long-Version}

\appendix[Proofs of Theorems and Corollaries]

\begin{proof}[Proof of Theorem~\ref{theorem1}]
If a vector ${\boldsymbol{\lambda}=(\lambda_1,\dots,\lambda_k)}$ is in the service rate region of this storage system, there exist $\lambda_{i,j}$'s, for ${i\in [n]}$ and ${j\in [t_i]}$, satisfying the set of constraints in~\eqref{eq:1a}, \eqref{eq:1b} and \eqref{eq:pos}. Based on the definition of the graph representation of codes and the fact that ${\mu_l=1}$, ${l\in [n]}$, it is easy to observe that the set of constraints in~\eqref{eq:1b} and~\eqref{eq:pos} are equivalent to the set of constraints in~\eqref{FracMatch}. Thus, the vector ${\boldsymbol{x}=(\lambda_{i,j}:i\in [k]~\text{and}~j\in [t_i])}$ is a fractional matching in the graph representation of $[n,k]_q$ code. Now, assume that a vector ${\boldsymbol{x}=(\lambda_{i,j}:i\in [k]~\text{and}~j\in [t_i])}$ is a fractional matching in the graph representation of $[n,k]_q$ code. Hence, the vector $\boldsymbol{x}$ satisfies the sets of constraints in~\eqref{FracMatch}, or equivalently, it satisfies the set of constraints in~\eqref{eq:1b} and~\eqref{eq:pos}. Based on Definition~\ref{def:SerRateReg}, a vector ${\boldsymbol{\lambda}=(\lambda_1,\dots,\lambda_k)}$ obtained from $\boldsymbol{x}$ using~\eqref{eq:1a} is in the service rate region of ${[n,k]_q}$ code. 
\end{proof}

\begin{proof}[Proof of Corollary~\ref{corollary1}]
An instance of the maximum fractional matching vector in the graph representation of an $[n,k]_q$ code can be obtained by solving the following LP according to~\eqref{MaxFracMatch}.
\begin{align*}
&\max~~~~{\sum_{i=1}^{k}\sum_{j=1}^{t_i}\lambda_{i,j}}\\
&~\text{s.t.}~~~~~~~\eqref{eq:1b},~\eqref{eq:pos}\nonumber
\end{align*}

According to the Theorem~\ref{theorem1}, there exist a demand vector ${\boldsymbol{\lambda}=(\lambda_1,\dots,\lambda_k)}$ in the service rate region which is obtained from $\boldsymbol{x}^\star=(\lambda_{i,j}^\star:i\in [k]~\text{and}~j\in [t_i])$ using~\eqref{eq:1a}. We want to show that the vector $\boldsymbol{\lambda}$ is in fact a  maximum demand vector ${\boldsymbol{\lambda}}^\star$. Using~\eqref{eq:1a}, we have $\sum_{i=1}^{k}\sum_{j=1}^{t_i}\lambda_{i,j}^\star=\sum_{i=1}^{k}\lambda_i$. Thus, it can be easily verified that $\boldsymbol{x}^\star$ provides a solution for the following LP:
\begin{align*}
&\max~~~~{\sum_{i=1}^{k}\lambda_i}\\
&~\text{s.t.}~~~~~~\eqref{eq:1a},~\eqref{eq:1b},~\eqref{eq:pos}\nonumber
\end{align*}

Based on~\eqref{maxrate}, an instance of the maximum demand vector is obtained by solving the above linear programming. Thus, the vector ${\boldsymbol{\lambda}}$ resulted from $\boldsymbol{x}^\star$ by~\eqref{eq:1a} is a maximum demand vector ${\boldsymbol{\lambda}}^\star$. On the other hand, for a maximum demand vector $\boldsymbol{{\lambda}}^\star$ in the service rate region which is obtained from~\eqref{maxrate}, there exists a fractional matching vector $\boldsymbol{x}$ that based on a similar reasoning, provides a solution for~\eqref{MaxFracMatch}. Thus, the vector $\boldsymbol{x}$ is a maximum fractional matching vector $\boldsymbol{x}^\star$ in the graph representation of $[n,k]_q$ code. 
\end{proof}

\begin{proof}[Proof of Theorem~\ref{theorem2}] 
According to Corollary~\ref{corollary1}, a maximum demand vector ${\boldsymbol{{\lambda}}^\star=(\lambda_{1}^\star,\dots,\lambda_{k}^\star)}$ is obtained from a maximum fractional matching vector ${\boldsymbol{x}^\star=(\lambda_{i,j}^\star:i\in [k]~\text{and}~j\in [t_i])}$ using~\eqref{eq:1a}. It follows that
\[\sum_{i=1}^{k}\lambda_{i}^\star=\sum\lambda_{i,j}^\star\] 
where $\lambda_{i,j}^\star$'s are the elements of ${\boldsymbol{x}^\star}$. Thus, ${\lambda^\star(\mathbf{G})=m_f(G)}$. Thus, according to Proposition~\ref{prop:duality}, we have ${{m(G)\leq \lambda^\star(\mathbf{G})\leq v(G)}}$.
\end{proof}

\begin{proof}[Proof of Theorem~\ref{theorem3}] 
Based on the definition of the bipartite graph, a graph ${G(V,E)}$ is a bipartite graph if the vertices $V$ of the graph, can be divided into two disjoint and independent sets, say $V_1$ and $V_2$ such that every edge of the graph ${e \in E}$ connects a vertex in $V_1$ to one in $V_2$. Thus, in order to show that the graph representation of the $k$-dimensional binary simplex code with generator matrix $\mathbf{G}$ is a bipartite graph, we need to determine the two disjoint sets of vertices, i.e., $V_1$ and $V_2$, in the graph representation $G(V,E)$ of the $[2^k-1,k]_2$ simplex code. Then, we have to prove that every edge ${e \in E}$ of the graph representation connects a vertex in ${V_1}$ to one in ${V_2}$ or equivalently we have to prove that there is no edge between the vertices in ${V_1}$ as well as in ${V_2}$. 

The set of vertices $V$ of the graph representation ${G(V,E)}$ of the ${[2^k-1,k]_2}$ simplex code correspond to the files stored on the storage nodes or the columns of the generator matrix $\mathbf{G}$. The columns of the generator matrix $\mathbf{G}$ of the $[2^k-1,k]_2$ simplex code are the set of all non-zero vectors of $\mathbb{F}_2^k$. Note that up to column permutations the generator matrix $\mathbf{G}$ of the ${[2^k-1,k]_2}$ simplex code is unique. Now, we can partition the columns of $\mathbf{G}$ into two sets $V_1$ and $V_2$ such that $V_1$ is the set of all non-zero column vectors in ${\mathbb{F}_2^k}$ with odd number of ones and $V_2$ is the set of all non-zero column vectors in ${\mathbb{F}_2^k}$ with even number of ones. Thus, $V_1$ and $V_2$ are two disjoint sets of columns that partition the columns of $\mathbf{G}$. Moreover, the self-loops corresponding to the systematic recovery sets are removed from the graph representation by adding dummy nodes. Consider each dummy node (column) as a zero vector in ${\mathbb{F}_2^k}$, denoted by $\mathbf{0}$. Thus, $V_1$ and $V^\prime_2=\{V_2 \cup \mathbf{0}\}$ determine two disjoint sets of vertices partitioning $V$ in the ${G(V,E)}$.

We want to prove that there is no edge between the vertices corresponding to the set of columns $V_1$ and there is no edge between the vertices corresponding to the set $V^\prime_2$. The proof is based on the contradiction approach. Let ${\mathbf{x},\mathbf{x}' \in V_1}$. Assume that there is an edge between the vertices corresponding to the ${\mathbf{x},\mathbf{x}' \in V_1}$. This means that $\{\mathbf{x},\mathbf{x}'\}$ forms a recovery set for a file $f_i$, $i \in [k]$, i.e., ${\mathbf{x}+\mathbf{x}'=\mathbf{e}_i}$. Since both $\mathbf{x}$ and $\mathbf{x}'$ have an odd number of ones, their sum must have an even number of ones which is a contradiction. Thus, there is no edge between the vertices in $V_1$. The proof for $V^\prime_2$ is similar. Let ${\mathbf{x},\mathbf{x}' \in V^\prime_2}$. Since both $\mathbf{x}$ and $\mathbf{x}'$ have an even number of ones, their sum must have an even number of ones which shows that $\{\mathbf{x},\mathbf{x}'\}$ cannot be a recovery set for any file $f_i$, $i \in [k]$. Thus, there is no edge between any ${\mathbf{x},\mathbf{x}' \in V^\prime_2}$. 
\end{proof}

\begin{proof}[Proof of Corollary~\ref{cor:simplex}]
The proof consists of two parts. First, we need to prove the converse by showing that the service capacity cannot be bigger than $2^{k-1}$. It can be easily seen that the set of all $2^{k-1}$ vertices corresponding to the columns of $\mathbf{G}$ with odd number of ones forms a minimum vertex cover in the graph representation of the ${[2^k-1,k,2^{k-1}]_2}$ simplex code. The reason is that since the graph representation of the this code, based on Theorem~\ref{theorem3}, is a bipartite graph, all the edges are covered by either one of the two partitions, i.e., $V_1$ and $V^\prime_2$ introduced in the proof of Theorem~\ref{theorem3}. Thus, the vertex cover number of this graph is $v(G)=2^{k-1}$ which indicates that $\lambda^\star(\mathbf{G})\leq 2^{k-1}$. 
\\

Next, we show the achievability proof by showing that there exists a vector ${{\boldsymbol{\lambda}}=(\lambda_{1},\dots,\lambda_{k})}$ in service rate region of the $[2^k-1,k,2^{k-1}]_2$ simplex code with $\sum_{i=1}^k\lambda_i=2^{k-1}$. For this purpose, since the graph representation of this code is a bipartite graph, we have $m(G)=v(G)=2^{k-1}$ which means that there exists a matching of size ${2^{k-1}}$ in the graph representation of the binary $k$-dimensional simplex code. For the $[2^k-1,k,2^{k-1}]_2$ simplex code which is a special subclass of availability codes, it is known that every file $f_i$ for ${i \in [k]}$ can be recovered from ${2^{k-1}-1}$ (availiability) disjoint groups of two (locality) servers. Thus, by considering the systematic recovery set, for each file $f_i$, $i \in [k]$, there are exactly $2^{k-1}$ disjoint recovery sets. One can consider the set of edges ${\{\lambda_{i,1},\dots,\lambda_{i,2^{k-1}}\}}$, for every ${i \in [k]}$, as an instance of matching in the graph representation. Corresponding to this matching, a demand vector ${{\boldsymbol{\lambda}}=2^{k-1}.\mathbf{e}_i}$ for ${i \in [k]}$ is obtained by applying~\eqref{eq:1a}. 
\end{proof}

\begin{proof}[Proof of Corollary~\ref{cor:integraleq}]
The proof is similar to the proof of Theorem~\ref{theorem1}.
\end{proof}

\begin{proof}[Proof of Corollary~\ref{cor:integralmaxeq}]
The proof is similar to the proof of Corollary~\ref{corollary1}.
\end{proof}

\begin{proof}[Proof of Theorem~\ref{theorem4}]
The existence of all vectors in set  $S_t=\{\boldsymbol{\lambda}=(\lambda_1,\dots,\lambda_k)|\sum_{i}^{k}\lambda_i=t, \lambda_i \in \mathbb{Z}_{\geq 0}\}$ in the integral service rate region $\mathcal{S}_I(\mathbf{G})$ of code $\mathbf{G}$ indicates that for any multiset of indices $\{i_1,\dots,i_t\}$, $i_j \in [k]$, the requests for the information symbols $f_{i_1},f_{i_2},\dots, f_{i_t}$ can be served at the same time by the storage system. On the other hand, each server can serve up to one request at a time, i.e., $\mu_l=1$ for all servers $l \in [n]$, which shows that $\lambda_{i,j}$ are binary. As a result, $t$ disjoint recovery sets are used for satisfying each demand vector $\boldsymbol{\lambda} \in S_t$. This means that for every multiset of indices $\{i_1,\dots,i_t\}$, $i_j \in [k]$, there exist $t$ disjoint sets $T_1,\dots,T_t$ of indices of columns in the generator matrix $\mathbf{G}$ such that for each $j \in [t]$, there exists a linear combination of columns of $\mathbf{G}$ indexed by $T_j$ which equals to the vector $\mathbf{e}_{i_j}$. Therefore, based on Proposition~\ref{batchcode}, the code $\mathbf{G}$ is a $[n,k,t]_q$ batch code.
\end{proof}

\begin{proof}[Proof of Corollary~\ref{cor:genPIR}]
The existence of the set  $S_t=\{t.\mathbf{e}_1=(t,0,\dots,0), \dots, t.\mathbf{e}_k=(0,\dots,0,t)| t \in \mathbb{N}\}$ in the integral service rate region $\mathcal{S}_I(\mathbf{G})$ of code $\mathbf{G}$ indicates that for every $i \in [k]$, $t$ requests for file $f_{i}$ can be served at the same time by the storage system. Since $\mu_l=1$ for all servers $l \in [n]$ and $\lambda_{i,j}$ are binary, one can readily confirm that for each file $f_i$, $i \in [k]$, there exist $t$ disjoint recovery sets which are used for satisfying the demand vector $t.\mathbf{e}_i \in S_t$. Thus, for every $i \in [k]$, there exist $t$ disjoint sets of columns in the generator matrix $\mathbf{G}$ that add up to $\mathbf{e}_{i}$. Therefore, based on definition~\ref{PIR}, the code $\mathbf{G}$ is an $[n,k]_2$ t-server PIR code.
\end{proof}
\end{document}